\documentclass[10pt]{llncs}

\RequirePackage{etex} 
\usepackage{graphicx}
\usepackage[colorlinks,linktocpage]{hyperref}
\hypersetup{
   colorlinks   = true, 
   urlcolor     = blue, 
   linkcolor    = blue, 
   citecolor    = green 
}

\usepackage{float} 
\restylefloat{figure}

\usepackage{mathtools}
\usepackage{url}
\usepackage{multicol}
\usepackage{multirow}

\usepackage[
	lambda,
	operators,
	advantage,
	sets,
	adversary,
	landau,
	probability,
	notions,	
	logic,
	ff,
	mm,
	primitives,
	events,
	complexity,
	asymptotics,
	keys]{cryptocode}

\newcommand{\rankweight}{\mathrm{w_R}}

\newcommand{\msg}{\mathsf{msg}} 
\newcommand{\cmt}{\mathsf{cmt}} 
\newcommand{\ch}{\mathsf{ch}} 
\newcommand{\rsp}{\mathsf{r}} 
\newcommand{\hlength}{\mathsf{h}}
\newcommand{\signature}{\mathsf{sgn}}
\newcommand{\sign}{\mathsf{Sign}} 
\newcommand{\verif}{\mathsf{Verify}} 
\newcommand{\drot}{\mathrm{drot}} 
\newcommand{\rot}{\mathrm{rot}} 
\newcommand{\seed}{\mathsf{seed}}
\newcommand{\xof}{\mathsf{XOF}}

\begin{document}
\title{Improved Veron Identification and Signature Schemes in the Rank Metric}
\titlerunning{Improved Veron Identification and Signature Schemes in the Rank Metric}
%
\author{Emanuele Bellini \inst{1} \and
Florian Caullery\inst{1} \and
Philippe Gaborit\inst{2} \and
Marc Manzano\inst{1} \and
Victor Mateu\inst{1}}
%
\authorrunning{E.~Bellini et al.}
%
\institute{Darkmatter LLC, Abu Dhabi, UAE 
\and
Universit\'e de Limoges, France
}
\maketitle              
\setcounter{tocdepth}{2}
%
\begin{abstract}
It is notably challenging to design an efficient and secure signature scheme based on error-correcting codes.
An approach to build such signature schemes is to derive it from 
an identification protocol through the Fiat-Shamir transform. 
All such protocols based on codes must be run several rounds, 
since each run of the protocol allows a cheating probability of either 2/3 or 1/2. 
The resulting signature size is proportional to the number of rounds, 
thus making the 1/2 cheating probability version more attractive.
We present a signature scheme based on double circulant codes in the rank metric, 
derived from an identification protocol with cheating probability of 2/3.
We reduced this probability to 1/2 to obtain the smallest signature among signature schemes based on the Fiat-Shamir paradigm, around 22 KBytes for 128 bit security level.
Furthermore, among all code-based signature schemes, our proposal has the lowest value of signature plus public key size, and the smallest secret and public key sizes.
We provide a security proof in the Random Oracle Model, implementation performances, and a comparison with the parameters of the most important code-based signature schemes.

\end{abstract}
\section{Introduction}
\label{sec:intro}
Due to the early stage of post-quantum algorithm research,
it is of paramount importance to provide the full range of 
quantum secure cryptographic primitives (signatures, key exchange, etc.) 
for all the main mathematical problems cryptography relies on.
This way, it will be easier to switch from one scheme to the other 
in the case one of the problems turns out to be insecure in the quantum model.
Given that it is the oldest quantum resistant family and, hence, 
the most thoroughly studied among all the contenders, 
code-based cryptography is a strong candidate in the NIST competition to standardize quantum resistant cryptographic algorithms \cite{NISTRound1}.
 
This work focuses on code-based cryptography digital signature schemes. 
Designing such schemes efficiently has been a grueling challenge and 
mainly three different approaches have been followed, with very little success. 
Hash-and-sign was introduced in pioneering work 
of Courtois, Finiasz, and Sendrier \cite{courtois2001achieve},
and is probably the most popular approach of the three.
It is based on the existence of a trapdoor which allows fast decoding, 
obtained by hiding a structured code into a random linear code.
Different choices of the underlying code lead to different instantiations of the scheme. 
All hash-and-sign schemes yield to small signatures (few thousands bits), but large public keys (order of MBytes), in some cases even non-practical ones for 128 bit security level and above.
Furthermore, almost all these schemes have been attacked.
The other two approches avoid the use of trapdoors. The first is usually referred to as the KKS (Kabatianskii-Krouk-Smeets) signature scheme \cite{kabatianskii1997digital}, who later evolved in the BMS (Barreto-Misoczki-Simplicio) scheme \cite{barreto2011one}.
Both of them can be instantiated on top of general linear codes. 
KKS and BMS have a good balance between public key (few tens of thousands of bits) and signature size (few thousands of bits), but they can only be considered one-time signature schemes. 
The third approach uses the Fiat-Shamir transform to turn a zero-knowledge identification scheme into a signature scheme, as initially proposed by Stern \cite{stern1993new} in 1993. 
The main drawback of such scheme is the large signature size. 
Many researchers followed Stern approach, trying to improve either the signature or the key size of the scheme.


In this manuscript, we provide a variation of a signature scheme based on Stern approach, obtaining the smallest signature (sgn), secret and public key (pk) sizes in the literature.
Compared to other approaches used to build code-based signature schemes,
we also have the smallest $|\signature| + |\pk|$ value.
We derive such signature from a 5-pass identification protocol with cheating probability 1/2.
We provide a security proof in the Random Oracle Model,
a detailed pseudo-code, implementation performances,  
set of parameters for 80, 128, 192, and 256 bit of classical security, and 
a comparison with the parameters of the most important code-based signature schemes.

The paper is organized as follows: Sect.~\ref{sec:rel} provides an overview of previous works
and the ideas behind the scheme. 
In Sect.~\ref{sec:prel} we provide the notions that are needed to understand the contribution.
Sect.~\ref{sec:id_veron_dc} presents our new identification protocol. 
Sect.~\ref{sec:param} sets the parameters of our signature schemes. 
Sect.~\ref{sec:perf} argues about the theoretical complexity of 
the key generation, signature and verification algorithms, 
providing also implementation details and performances. 
Sect.~\ref{sec:comp} shows a comparison of the parameters of our proposal and other well-known code-based signature schemes, and Sect.~\ref{sec:concl} draws the conclusions.
%
%
\section{Main idea}
\label{sec:rel}
Commonly, cryptographic signature schemes whose security relies on the difficulty of decoding a linear code are built by converting a 3 or a 5-pass identification protocol into a signature scheme via the Fiat-Shamir transform or a generalization of it.
The first to propose such paradigm was Stern \cite{stern1993new}.
In this work, Stern exhibits a 3-pass identification protocol
whose security is based on the difficulty of decoding a random linear code  and finding a hash collision,
and in which a cheater can correctly identify with a probability of 2/3.
For this last reason, the protocol should be run an appropriate number of rounds which depends on the security level the scheme needs to reach.
Since the corresponding signature is proportional to the number of rounds, this means that this type of approach yields to large signatures, of the order of hundreds of KBytes.
The basic idea of the protocol is that, given the parity-check matrix $H$ of a linear code as a public parameter, a random vector $e$ of weight $w$, and a public key $s=eH$, 
the prover needs to prove the knowledge of two properties, 
namely the fact that
the vector $e$ is generating the syndrome $s$, and
that $e$ has Hamming weight $w$.
Adding a random commitment there are always two possibilities for
cheating among the three cases.
In the same work, Stern shows how to reduce the cheating probability to 1/2, 
by splitting the challenge step into two challenges, 
the second of which adds a \emph{variation} on $e$, 
forcing the protocol to perform 5 passes.
Precisely, $e$ was chosen as a codeword of a Reed-Muller code. 
Such trick allows to almost halve the corresponding signature size, 
even though, with this particular solution, there is a loss in efficiency.
Stern signature schemes presents very small secret keys (less than a thousand bits) and medium size public keys (one hundred thousand bits). 

Subsequent works aim at improving either key or signature sizes,
by (1) choosing a structured code rather than a random linear code, 
(2) changing the variation performed on $e$, 
(3) working with the dual cryptosystem, or 
(4) working in a different metric.
In \cite{veron1997improved}, Veron presents the dual of the 3-pass Stern proposal, 
i.e. it uses the generator matrix $G$ of a code, instead of the parity-check matrix as a public parameter, and uses a pair $(x,e)$ as a secret key, and a codeword $y=x G +e$ as a public key. This allows to send less data on average during the response step, implying slightly shorter signatures.
Later, Cayrel-Veron-El Yousfi Alaoui (CVE) \cite{cayrel2010zero} presented a 5-pass identification protocol with cheating probability of 1/2, 
using codes over $\FF_{2^m}$, rather than $\FF_2$ as done by Stern and Veron, 
and a scalar multiplication as the variation of $e$. 
In \cite{dagdelen2016extended} it is shown how to extend the Fiat-Shamir transform to a $n$-pass protocol (with $n$ odd).
In 2011, Gaborit, Schrek and Z\'emor \cite{gaborit2011full} 
presented the rank metric version of the Stern identification protocol, 
decreasing significantly key and signature sizes, 
due to the fact that rank metric decoding has quadratic exponential complexity, 
while Hamming metric decoding is linear exponential.
The same year, Aguilar, Gaborit and Schrek \cite{aguilar2011new}, 
used double circulant codes in the Hamming metric 
to reduce the key size of the Veron scheme, and 
presented a 5-pass version of it, with cheating probability close to 1/2,
 performing a variation of $e$ with a circulant rotation of its two halves in the second challenge step. 
Furthermore, they introduced a compression technique to reduce the signature size.
Recently, in \cite{bellini2018code}, a rank metric version of Veron and CVE has been presented, though lacking a security proof.
We are not aware of any attack to any of the Fiat-Shamir paradigm constructions, which probably have not received much attention from the cryptographic community yet.

In this work, we present a rank metric version of the 5-pass Veron double circulant signature scheme of \cite{aguilar2011new}, with a new variation performed on $e$, 
which allows us to reach a cheating probability much closer to 1/2.
Precisely, we adopt a random linear combination of all possible rotations of $e$
in the second challenge step.
We also present a compressed version of the scheme, which achieves signature sizes that are comparable to the one of post-quantum hash-based signature schemes.

\section{Preliminaries}
\label{sec:prel}

%
%
In this section we provide the essential definition of the objects that are used in our protocol.

A linear $(n,k)_q$-code $C$ is a vector subspace of $(\FF_q)^n$ of dimension $k$, where $k$ and $n$ are positive integers such that $k<n$, $q$ is a prime power, and $\FF_q$ is the finite field with $q$ elements.
Elements of the vector space are called vectors or words, while elements of the code are called codewords.
A matrix $G \in \FF_q^{k \times n}$ is called a generator matrix of $C$ if its rows form a basis of $C$, i.e.
$C = \{x\cdot G : x \in (\FF_q)^k\}$.
A matrix $H \in \FF_q^{(n-k) \times n}$ is called a parity-check matrix of $C$ if 
$C = \{x\in (\FF_q)^n: H\cdot x^T = 0\}$.
Our schemes will use a special type of linear codes, called \emph{double circulant} codes,
which are a special case of \emph{quasi-cyclic} (or \emph{circulant}) codes (see e.g. \cite{misoczki2013mdpc}).
\begin{definition}[Double Circulant Codes]
	Let $n = 2k$ for an integer $k$. 
	Consider a vector $x = (x_1, x_2)$ of $(\FF_{q})^n$ as a pair of two blocks of length $k$.
	An $[n, k]$ linear code $C$ is \emph{Double Circulant} (DC) if, for any $c = (c_1, c_2) \in C$, 
	the vector obtained after applying a simultaneous circular shift to both blocks $c_1, c_2$ 
	is also a codeword.
	More formally, by considering each block $c_1, c_2$ as a polynomial in $R = \FF_{q}[X]/(X^n - 1)$, 
	the code $C$ is DC if for any $c = (c_1, c_2) \in C$ it holds that $(X \cdot c_1, X \cdot c_2) \in C$.
	
	A \emph{systematic} double circulant $[n,k]$ code 
	is a double circulant code with a parity-check matrix of the form $H = [I_k | A]$,
	where $I_k$ is the identity matrix of size $k$, and 
	$A$ is a $k \times k$ circulant matrix.
\end{definition}
In this paper we work with codes in the \emph{rank metric}. 
Given a fixed basis $b = \{b_1, \ldots, b_m\}$ of $(\FF_q)^m$, 
a vector $a \in (\FF_{q^m})^n$ can be represented as a matrix with entries in $\FF_q$, by expanding each component of $a_i$  with respect to $b$ in a column $(a_{1,i}, \ldots, a_{m,i})^T$. 
where $a_i = \sum_{j=1}^{m} a_{j,i} b_j, i = 1, \ldots, n$. 
We define the rank of a vector as the rank of its \emph{matrix representation},
with respect to $b$. 
We denote the previous matrix representation as $\phi_b(a)$, and 
by $\phi_b^{-1}$ the inverse map. 
In what follows, we will omit $b$ as we consider it fixed.

To send a binary vector of a certain Hamming weight 
to \emph{any} other vector of the same Hamming weight, 
it is sufficient to apply a random permutation to vector components.
The map with the analogue property in the rank metric, 
i.e. sending a vector of a certain  rank to \emph{any} other vector of the same rank, 
can be defined as follows (see \cite{gaborit2011full}).
\begin{definition}
	Let 
	$Q \in M_{m,m}(\FF_q)$ be a $q$-ary matrix of size $m \times m$, 
	$P \in M_{n,n}(\FF_q)$ be a $q$-ary matrix of size $n \times n$, and
	$v \in (\FF_{q^m})^n$.  
	We define the function $\Pi_{P,Q}$ such that $\Pi_{P,Q}(v) = \phi^{-1}(Q \cdot \phi(v) \cdot P)$, i.e.
	\begin{align*}
	\Pi_{P,Q} : (\FF_{q^m})^n & \mapsto  (\FF_{q^m})^n \\
	(v_1,\ldots,v_n) & \mapsto (\pi_1, \ldots, \pi_n)
	\end{align*}
	where for $h=1,\ldots,n$, \\
	$
	\pi_h:= \beta_1 \sum_{i=1}^m \sum_{j=1}^n Q_{1,i} v_{i,j} P_{j,h} + \ldots + \beta_m \sum_{i=1}^m \sum_{j=1}^n Q_{m,i} v_{i,j} P_{j,h}
	$
\end{definition}
%
It is proved in \cite{gaborit2011full} that
the following properties hold for $\Pi_{P,Q}$.
\begin{itemize}
	\item For any $x, y \in (\FF_{q^m})^n, P \in M_{n,n}(\FF_q)$ and $Q \in  M_{m,m}(\FF_q)$ then:
	\begin{itemize}
		\item (rank preservation) $\rankweight(\Pi_{P,Q}(x)) = \rankweight(x)$;
		\item (linearity) $a \Pi_{P,Q}(x) + b \Pi_{P,Q}(y) = \Pi_{P,Q}(ax+by)$.
	\end{itemize}
	\item For any $x, y \in (\FF_{q^m})^n$ such that $\rankweight(x)=\rankweight(y)$, it is possible to find $P \in M_{n,n}(\FF_q)$ and $Q \in  M_{m,m}(\FF_q)$ such that $x = \Pi_{P,Q}(y)$. 
\end{itemize}




%
Both in the Hamming and in the rank metric, 
random codes over $\FF_q$ asymptotically achieve the Gilbert-Varshamov bound~\cite{gabidulin1985theory}.
Furthermore, they have close to optimal correction capability \cite{loidreau2006properties}.

We now define the problems upon which the security of the schemes we present is based.

\begin{definition}[RSD Distribution]
	Given the positive integers $n, k$, and $r$, the $RSD(n, k, r)$ Distribution chooses
	$H \sample (\FF_{q^m})^{(n-k) \times n}$ and $x \sample (\FF_{q^m})^{n}$
	such that $\rankweight(x) = r$, and outputs $(H,  H \cdot x^T)$
\end{definition}

\begin{problem}[RSD Problem]
	On input $(H, y^T) \in (\FF_{q^m})^{(n-k) \times n} \times (\FF_{q^m})^{n}$ from the RSD distribution, 
	the Rank Syndrome Decoding problem RSD($n, k, r)$ asks to find $x \in  (\FF_{q^m})^n$
	such that $H \cdot x^T = y^T$ and $\rankweight(x) = r$.
\end{problem}
The previous problem can be defined correspondingly also in the Hamming metric,
in which setting the problem has been proven to be NP-complete \cite{berlekamp1978inherent}.
The RSD problem has recently been proven difficult with a probabilistic reduction to
the Hamming scenario in \cite{aguilar2016efficient}. 
For cryptography, it is also useful to use the Decisional version of the problem.
Our scheme security depends on the difficulty of solving the same RSD problem 
defined with Double Circulant codes, rather than random linear codes.
The decisional version of this problem is a special case of 
the Decisional Rank $s$-Quasi Cyclic Syndrome Decoding Problem 
defined for example in \cite{aguilar2016efficient}.
There is no known reduction from the search version of this problem to its decisional version. 
However, the best known attacks on the decisional version of the problem 
remain the direct attacks on the search version of the problem.

\section{Veron Double Circulant identification protocol in the rank metric}
\label{sec:id_veron_dc}


The scheme we present in this section,
to which we refer to as the Rank Veron Double Circulant (RVDC) identification protocol, 
mixes the ideas from \cite{gaborit2011full}, 
where the Stern protocol is converted from Hamming to rank metric and the function $\Pi_{P,Q}$ (see Section \ref{sec:prel} above) is introduced,
and from \cite{aguilar2011new}, 
where the cheating probability of the Veron protocol 
is improved from 2/3 to 1/2 
using the double circulant technique in the Hamming metric.
In \cite{aguilar2011new}, the intermediate challenge is a random parallel left rotation.
To better exploit the rank metric properties, and to make it more difficult to guess the challenge for an attacker, we instead consider a random linear combination of all possible parallel left rotations.
\begin{definition}
	Let $n=2k$
	and $x = (x_1, \ldots, x_{k}) \in (\FF_{q^m})^{k}, y = (y_1, \ldots, y_{n}) \in (\FF_{q^m})^{n}$. 
	We denote with 
	\begin{align*}
	\rot_{i}((x_1, \ldots, x_k))= 
	(x_{i+1}, \ldots,  x_{k}, x_{1}, \ldots, x_{i})
	\end{align*}
	the left rotation of $i$ positions of the vector $x$,
	and with 
	\begin{align*}
		\drot_{i}((y_1, \ldots, y_{i},y_{i+1}, \ldots, y_k,y_{k+1}, \ldots, y_{k+i}, y_{k+i+1}, \ldots, y_{k+k} ))= \\
		(y_{i+1}, \ldots,  y_{k}, y_{1}, \ldots, y_{i},y_{k+i+1}, \ldots, y_{k+k}, y_{k+1}, \ldots, y_{k+i} )
	\end{align*}
	the parallel left rotation of $i$ positions of the two halves of the vector $y$.
	Given
	$a = (\alpha_1, \ldots, \alpha_{k}) \in (\FF_{q})^{k}$
	we also denote with $\Gamma'_{a}(x)$ the linear combination of all possible $k$ left rotations of $k-i$ positions of $x$, and $\Gamma_{a}(y)$ the linear combination of all possible $k$ parallel left rotations of $i$ positions of $y$
	\begin{align*}
		\Gamma'_{a}(x) = \sum_{i=1}^{k}  \alpha_i  \cdot \rot_{k-i}(x) \in (\FF_{q^m})^k \,,  \quad 
		\Gamma_{a}(y)  = \sum_{i=1}^{k}  \alpha_i  \cdot \drot_{i}(y) \in (\FF_{q^m})^n \,.
	\end{align*}
\end{definition}

The following lemma, used to prove the completeness of the scheme, can be easily proven.
\begin{lemma}
	\label{thm:gamma_property}
Given the $k \times 2k$ generator matrix $G$ of a double circulant linear code
and a vector $x = (x_1, \ldots, x_{k}) \in (\FF_{q^m})^{k}$,
the following property holds
\begin{align*}
	\Gamma_a(x \cdot G) = \Gamma'_a(x) \cdot G 
\end{align*}
\end{lemma}

As we already noted in Section \ref{sec:prel} a codeword $y$ of a $[2k,k]$ double circulant code can be seen as the concatenation of two blocks, i.e.  $y=(y_1, y_2)$, of length $k$.
If we consider each block $y_1, y_2$ as a polynomial in $R = \FF_{q}[X]/(X^k - 1)$ then the function $(y_1,y_2) \mapsto \drot_{i}((y_1,y_2))$ is equal to $(y_1,y_2) \mapsto (X^i \cdot y_1,X^i \cdot y_2)$, where the multiplication by $X^i$ is performed in the ring, i.e. modulo $X^k - 1$.

Although there is no general complexity result for quasi-cyclic codes, 
their decoding is considered to be difficult by the community. 
There exist structural attacks which uses the cyclic structure of 
the code \cite{sendrier2011decoding,hauteville2015new,guo2015new,londahl2016squaring}, 
but these attacks have only a very limited impact  on the practical complexity of the problem. 
These attacks are especially efficient in the case when 
the polynomial $X^n - 1$ has many small factors.
These attacks become inefficient as soon as $X^n - 1$ has only two factors 
of the form $(X - 1)$ and $X^{n-1} + X^{n-2} + \ldots + X + 1$,
which is the case when $n$ is primitive in $\FF_{q^m}$.
The conclusion is that in practice, 
the best attacks are the same as those for non-circulant
codes up to a small factor.
Another solution to completely avoid such attacks is to use the ring
$R = \FF_{q}[X]/(X^k - p(X))$, where $p(X)$ is a polynomial with coefficients in $\FF_{q}$, and $X^k - p(X)$ is irreducible over $\FF_{q}$.

Recall that we will denote by $\secpar$ the security level of the scheme.
The key generation algorithm is listed in Fig. \ref{fig:verondb_keygen}.
The RVDC identification protocol is listed
in Fig. \ref{fig:verondb_idp}.

\restylefloat*{figure}
\begin{figure}[htb]
\centering
\procedure{RVDC: \kgen(\secparam)}{%
\pcln \text{Define } m,n,k,r \text{ as in Sect. \ref{sec:param}}\\
\pcln x \sample (\FF_{q^m})^k \\
\pcln e \sample (\FF_{q^m})^n  \text{ s.t. } \rankweight(e) = r\\
\pcln \sk \gets (x,e) \\
\pcln G \sample (\FF_{q^m})^{n}  \\
\pcln G' \in (\FF_{q^m})^{k \times n} \gets \text{ Expand } G \text{ in double circulant form} \\
\pcln y \gets x\cdot G' + e \\
\pcln \pk \gets (y,G,r)\\ 
\pcln \pcreturn \sk, \pk
}
\caption{RVDC key generation algorithm in the rank metric}
\label{fig:verondb_keygen}
\end{figure}

\begin{figure}[ht]
	\centering
	\pseudocode[codesize=\scriptsize]{%
		\textbf{Prover} \< \< \textbf{Verifier} \\[0.1\baselineskip][\hline]
		\< \< \\[-0.5\baselineskip]
		\sk,\pk = (x,e),(y,G,r)\gets \kgen \< \< \pk \\[0.1\baselineskip][\hline]
		\< \< \\[-0.5\baselineskip]
		u \sample (\FF_{q^m})^k \\
		Q \sample M_{m,m}(\FF_{q}) , P \sample M_{n,n}(\FF_{q}) \< \< \\ \< \< \\
		c_1 \gets \hash(P, Q) \<  \< \\
		c_2 \gets \hash(\Pi_{P,Q}(u \cdot G)) \< \sendmessageright*[1.5cm]{c_1,c_2} \< \\
		\<  \< a = (\alpha_1, \ldots, \alpha_{k}) \sample (\FF_{q})^{k}, \\
		\< \sendmessageleft*[1.5cm]{a} \< a_i \text{ not all the same} \\
		c_3 \gets \hash(\Pi_{P,Q}(u \cdot G +  \Gamma_{a}(e))) \< \sendmessageright*[1.5cm]{c_3} \< \\
		\< \sendmessageleft*[1.5cm]{b} \< b \sample \{0,1\} \\
		\pcif b=0 \< \< \\
		\pcind  \rsp_1 \gets (P, Q), \rsp_2 \gets u+ \Gamma'_{a}(x) \< \sendmessageright*[1.5cm]{\rsp_1,\rsp_2} \< \pcif c_1 = \hash(\rsp_1) \wedge \\ 
		\< \< \pcind c_3=\hash(\Pi_{\rsp_1}(\rsp_2\cdot G +  \Gamma_{a}(y)) \\
		\< \< \pcind[2] \pcreturn \true \\
		\pcif b=1 \< \< \\
		\pcind  \rsp_1 \gets \Pi_{P,Q}(u \cdot G),\rsp_2 \gets \Pi_{P,Q}( \Gamma_{a}(e)) \< \sendmessageright*[1.5cm]{\rsp_1,\rsp_2} \< \pcif c_2 = \hash(\rsp_1) \wedge \\ 
		\< \< \pcind  c_3=\hash(\rsp_1 + \rsp_2) \wedge \\
		\< \< \pcind   \rankweight(\rsp_2)=r \\
		\< \< \pcind[2] \pcreturn \true
	}
	\caption{RVDC identification protocol in the rank metric}
	\label{fig:verondb_idp}
\end{figure}

In Section~\ref{sec:sign_descr_rvdc}, 
we describe how to convert the identification protocol from Fig. \ref{fig:verondb_idp} into a signature scheme, 
to which we will refer to as \emph{Rank Veron Double Circulant (RVDC) Signature scheme}, 
using a generalization of the Fiat-Shamir transform, introduced in \cite{dagdelen2016extended}.
The signature size of the scheme can be reduced 
by applying the commitment compression technique used in \cite{aguilar2011new}.
We will call the scheme resulting from this variation  \emph{compressed Rank Veron Double Circulant} (cRVDC) scheme.

In Sect.~\ref{sec:id_veron_dc_zn} we prove that the identification protocol is complete, sound and that the communication leaks no information on the secret key.
The security of RVDC scheme is based on a variant of the Rank Syndrome Decoding problem, 
that we call \emph{Differential Rank Decoding Problem}, 
defined as Problem~\ref{prob:DRDP} in the same section.


\section{Parameters choice}
\label{sec:param}

In this section we provide a set parameters for 80, 128, 192, 256 bit of classical security, 
corresponding to 40, 64, 96, 128 bit of quantum security, the last three falling into 
category 1,  3, and 5 in the NIST post-quantum competition.

The best generic combinatorial attack to solve the RSD problem has a complexity of 
$
\bigO{(n-k)^3 m^3 q^{r\frac{(k+1)m}{n}-m}}
$
\cite{aragon2017improvement}. 
If $k \ge \left\lceil \frac{(r+1)(k+1)-(n+1)}{r} \right\rceil$,
an algebraic approach \cite{gaborit2016complexity} is also possible to recover the error in 
$
\bigO{r^3 k^3 q^{r\left\lceil \frac{(r+1)(k+1)-(n+1)}{r} \right\rceil}}
$ steps. 
Finally, to avoid specific Gr\"obner basis attacks, the condition $n > r(k+1)$ should hold.
We choose the values $m,n,k,r$ accordingly.
As far as it concerns post-quantum security, the author of \cite{loidreau2017new}, in line with \cite{bernstein2010grover}, presents some arguments showing that the post-quantum complexity of RSD is computed by square-rooting the exponential term in the classical complexity formula.


Recall that for the case of double circulant code we have to choose $n=2k$.
As suggested in \cite{stern1993new}, it is better to choose $r$ slightly below 
the theoretical distance $d$ provided by the Gilbert-Varshamov bound, 
in order to avoid possible small rank attacks similar to small weight codewords attack such as \cite{stern1988method}.
We choose $m$ to be prime, so to have no subfields of $\FF_{2^m}$, 
which in other cases leads to attacks.
We also need to choose the number of rounds $\delta$ in order to decrease the impersonation probability to our needs. 
As far as it concerns the identification protocols, the impersonation probability of one single round for RVDC is $p=\frac{q^k+\rho}{2q^k}$ with overwhelming probability. 
To reach a security level $l$ with an impersonation probability of $p$,
i.e. to compute the number of round $\delta$, 
we need to set $\delta = \log_{p} (1/2^l)$. 
This results in $\delta = 81, 129, 193, 257$,
corresponding to 80, 128, 192, 256 bit security level in the classical scenario.
In Table \ref{tab:param_rvdc}, we propose 4 sets of parameters, respectively for the 80, 128, 192, and 256 bit security level in the classical scenario, for both RVDC and cRVDC signature schemes.

For all the proposed parameters it holds the condition $k < \left\lceil \frac{(r+1)(k+1)-(n+1)}{r} \right\rceil$, so the algebraic attack of \cite{gaborit2016complexity} must be taken into consideration while evaluating the security.

In the table $A=r^3 k^3 q^{r\left\lceil \frac{(r+1)(k+1)-(n+1)}{r} \right\rceil}$, 
$B=(n-k)^3 m^3 q^{r\frac{(k+1)m}{n}-m}$, 
$C=r^3 k^3 q^{r\left\lceil \frac{(r+1)(k+1)-(n+1)}{2r} \right\rceil}$, 
$D=(n-k)^3 m^3 q^{r\frac{(k+1)m}{2n}-m}$, 
\begin{table}
	\begin{center}
		\resizebox{\textwidth}{!}{
		\begin{tabular*}{\textwidth}{c @{\extracolsep{\fill}}|cccccccc|cc|cc}
			\multicolumn{9}{c}{Parameters} & \multicolumn{2}{c}{Classic Attacks WF} & \multicolumn{2}{c}{Quantum Attacks WF}\\
			$\secpar$ & $q$ & $m$ & $n$ & $k$ & $r$ & $\rho$ & $\delta$ & $\hlength$ & $\log_2 A$ & $\log_2 B$ &  $\log_2 C$ & $\log_2 D$ \\
			\hline
			96 & 2 & 29 & 22 & 11 & 7 & 10 & 81 & 160 & 95.801 & 106.68 & 60.800 & 51.316 \\
			125 & 2 & 31 & 26 & 13 & 8 & 10 & 129 & 256 & 124.10 & 128.50 & 76.102 & 61.733 \\
			193 & 2 & 41 & 34 & 17 & 10 & 10 & 193 & 384 & 192.23 & 204.39 & 112.23 & 95.864 \\
			252 & 2 & 47 & 38 & 19 & 12 & 10 & 257 & 512 & 251.50 & 279.25 & 143.50 & 130.83 \\
		\end{tabular*}
	}
	\end{center}
	\caption{RVDC and cRVDC parameters.}
	\label{tab:param_rvdc}
\end{table}

\section{Key and signature size comparison}
\label{sec:comp}
%
\begin{table}[!ht]
\begin{center}
\resizebox{\textwidth}{!}{
\begin{tabular}{cccc|rrr}
$\secpar$ & Scheme & Metric & Scheme parameters & $|\signature|$ & $|\sk|$ & $|\pk|$ \\
\hline
&  &  & $(m,t,\delta,i)$ & & &  \\
$81$  & Parallel-CFS \cite{finiasz2010parallel} & Hamm. & $(20,8,2,3)$  &   294 & 20\,971\,680 & 167\,746\,560 \\
$80$  & Parallel-CFS \cite{finiasz2010parallel} & Hamm. & $(17,10,2,2)$ &   196 & 2\,228\,394 & 22\,253\,340 \\
\hline
&  &  & $(n,k,\omega,Q)$ & & & \\
177       & RaCoSS \cite{roy2017racoss} & Hamm. & $(2400,2060,48,0.07)$ & 4800 & 5\,760\,000 & 816\,000 \\ 
177       & RaCoSS(Compr.) \cite{roy2017racoss} & Hamm. & $(2400,2060,48,0.07)$ & 2436 & 1\,382\,400 & 816\,000 \\ 
\hline
 &  &  & $(q,m,n,k,d,t,t',r)$ & & &  \\
128 & RankSign I \cite{gaborit2014ranksign}  & Rank & $(2^{32},21,20,10,2,2,1,8)$ & 11\,008 &    540\,288 & 80\,640  \\
128 & RankSign II \cite{gaborit2014ranksign} & Rank & $(2^{24},24,24,12,2,2,2,10)$& 12\,000 &    652\,032 & 96\,768  \\
192 & RankSign III \cite{gaborit2014ranksign} & Rank & $(2^{32},27,24,12,2,3,1,10)$& 17\,280 & 1\,034\,208 & 155\,520 \\
256 & RankSign IV \cite{gaborit2014ranksign} & Rank & $(2^{32},30,28,14,2,3,2,12)$& 23\,424 & 1\,527\,360 & 228\,480 \\
\hline
&  &  & $(r,m,p,w)$ & & &  \\
128 & pqsigRM-4-12 \cite{lee2017pqsignrm} & Hamm. & $(4,12,16,1295)$ & 4\,224 & 27\,749\,002 & 2\,621\,788 \\
196 & pqsigRM-6-12 \cite{lee2017pqsignrm} & Hamm. & $(6,12, 8, 311)$ & 4\,224 &  19\,326\,902 & 3\,980\,860 \\
256 & pqsigRM-6-13 \cite{lee2017pqsignrm} & Hamm. & $(6,13,16,1441)$ & 8\,320 & 16\,777\,216 & 84\,020\,992 \\
\hline
&  &  & $(n,k,w,k_U,k_V)$ & & &  \\
128 & Wave \cite{debris2018wave} & Hamm. & (5172, 3908, 4980, 2299, 1609) & 8\,326 & na & 7\,840\,000 \\
\hline\hline
&  &  & $(q,n,k,w,\delta,\hlength)$ & & &  \\
$80$  & Stern \cite{alaoui2013code} & Hamm. & $(2,   768, 384, 76, 141, 160)$ & 908\,534 &    768 & 147\,846 \\ 
$80$  & Veron \cite{alaoui2013code} & Hamm. & $(2,   768, 384, 76, 141, 160)$ & 872\,438 & 1\,152 & 148\,230 \\ 
$80$  & CVE \cite{alaoui2013code} & Hamm. & $(2^8, 144, 72,  55, 80,  160)$ &    531\,539 & 1\,152 &  42\,053 \\ 
\hline
&  &  & $(n,k,i,w,\delta,\hlength)$ & & &  \\
68 &  Veron Double Circulant \cite{aguilar2011new} & Hamm. & (698,349,19,70) & 93\,000 & 700 & 1050 \\ 
\hline
&  &  & $(q,m,n,k,r,\delta,\hlength)$ & & &  \\
68 &  Rank Stern \cite{gaborit2011full} & Rank & (2,20,20,11,3,137,160) & na & 400 & 2160 \\ 
\hline
&  &  & $(q,m,n,k,r,\rho,\delta,\hlength)$ & & &  \\
96 &  RVDC & Rank & (2,29,22,11,7,10,81,160) & 157\,140 & 957 & 960 \\ 
96 & cRVDC & Rank & (2,29,22,11,7,10,81,160) & 84\,863 & 957 & 960 \\ 
125 &  RVDC & Rank & (2,31,26,13,8,10,129,256) & 334\,626 & 1\,209 & 1\,212 \\ 
125 & cRVDC & Rank & (2,31,26,13,8,10,129,256) & 179\,854 & 1\,209 & 1\,212 \\ 
193 &  RVDC & Rank & (2,41,34,17,10,10,193,384) & 832\,409 & 2\,091 & 2\,095 \\ 
193 & cRVDC & Rank & (2,41,34,17,10,10,193,384) & 440\,510 & 2\,091 & 2\,095 \\ 
252 &  RVDC & Rank & (2,47,38,19,12,10,257,512) & 1\,437\,915 & 2\,679 & 2\,683 \\ 
252 & cRVDC & Rank & (2,47,38,19,12,10,257,512) & 762\,935 & 2\,679 & 2\,683 \\ 
\hline\hline
 &  & & $(n, h, d, \log t, k, w)$ & & & \\
133 & SPHINCS$^+$-128s \cite{bern2017sphincsplus}  & - & (16, 64, 8, 15, 10, 16)  & 64\,640     & 512   & 256   \\
128 & SPHINCS$^+$-128f \cite{bern2017sphincsplus}   & - & (16, 60, 20, 9, 30, 16)  & 135\,808   & 512   & 256   \\
196 & SPHINCS$^+$-192s \cite{bern2017sphincsplus}  & - & (24, 64, 8, 16, 14, 16)  & 136\,512    & 768   & 384   \\
195 & SPHINCS$^+$-192f \cite{bern2017sphincsplus}   & - & (24, 66, 22, 8, 33, 16) & 285\,312   & 768   & 384   \\
255 &  SPHINCS$^+$-256s \cite{bern2017sphincsplus} & - & (32, 64, 8, 14, 22, 16)  & 238\,336   & 1\,024  & 512   \\
254 & SPHINCS$^+$-256f \cite{bern2017sphincsplus}   & - & (32, 68, 17, 10, 20, 16) & 393\,728   & 1\,024  & 512  \\
\hline
\end{tabular}
}
\end{center}
\caption{Comparison of keys and signature bit sizes between our proposals and the most popular code-based and hash-based signature schemes.}
\label{tab:comp_keys_pq}
\end{table}

In Table \ref{tab:comp_keys_pq} we report some key and signature bit sizes 
for other signature schemes based on codes. 
In particular, we report the results of hash-and-sign signature schemes such as
Parallel-CFS \cite{finiasz2010parallel},
the three NIST competitors for signatures based on codes, 
RankSign \cite{gaborit2014ranksign}, RaCoSS \cite{roy2017racoss}, and pqsigRM \cite{lee2017pqsignrm}, 
and Wave \cite{debris2018wave}, which has been proposed very recently. 
We also add the results from \cite{alaoui2013code} regarding the Hamming variants of Stern, Veron and CVE signature schemes,
one entry for the parameters proposed in \cite{aguilar2011new} for the double circulant version of Veron scheme in the Hamming metric, 
and one entry for the parameters proposed in \cite{gaborit2011full} for the rank version of Stern signature scheme.
As far as it concerns the latter, we remark that when the work was published,
results from \cite{gaborit2016complexity}, \cite{aragon2017improvement}, and \cite{loidreau2017new} 
were not known, so the security was believed to be 83 bits.
While for the parameters in \cite{aguilar2011new}, according the decoding complexity estimation of $2^{0.097n}$ given in \cite{may2015computing}, 
the security of the scheme is about 68 bits,
while in \cite{aguilar2011new} was claimed to be 81.
Recall also that for all three NIST competitors some attacks have been found, 
so either the parameters should be made larger or 
some modification of the scheme will be proposed in the future.

For completeness, we also report key and signature size of one of the most popular hash-based signature scheme, SPHINCS$^+$, introduced in \cite{bernstein2015sphincs}. The parameters that we consider are from the NIST submission document \cite{bern2017sphincsplus}. We can see that SPHINCS$^+$ has signatures and keys that are from 2 to 5 times smaller compared to cRVDC.

\section{Performance}
\label{sec:perf}
The cost of 
RVDC, and cRVDC key generation algorithm 
is dominated by the multiplication of a vector to the generator matrix.
Only one multiplication is needed to generate the public key, 
and this makes the key generation particularly fast.
On the other hand, 
the cost of signature and verification algorithms
are dominated by the number of rounds and the cost of the underlying hash function.
In particular, 
%
in the RVDC scheme (see Appendix \ref{sec:sign_descr_rvdc}),
$3\delta + 2$ and $2\delta+2$ hashes have to be computed, respectively, for the signature and the verification.
In the cRVDC scheme,
$3\delta + 3$ and $2\delta + 3$ hashes have to be computed, respectively, for the signature and the verification.

In Table \ref{tab:perf}, we report the performance of our scheme on a MacBook Pro equipped with a 2.9 GHz Intel Core i7 and a Huawei P20 Pro equipped with a Kirin 970 supporting ARMv8 instructions. The implementation is using AVX2 or NEON instructions sets for the finite field arithmetic but not on any other part of the code. The hash functions used are from the SHA2 family when the digest size matched the requirements and SHAKE256 when a longer output was needed. We also used AES-CTR-DRBG as a PRNG for random number generation. We compared our implementation with the optimized implementation of SPHINCS$^+$-SHAKE256 from SPHINCS$^+$ NIST submission package. As observed, our proposals outperform SPHINCS$^+$ in all cases. The table entries are in operations per second.

\begin{table}
	\begin{center}
		\begin{tabular}{lccccccc}
			& & \multicolumn{3}{c}{Macbook Pro} & \multicolumn{3}{c}{Huawei P20 Pro} \\
			Scheme & Security Level & $\;\;\kgen$ & $\;\;\sign$ & $\;\;\verify$ & $\;\;\kgen$ & $\;\;\sign$ & $\;\;\verify$ \\
			\hline
			RVDC   & 80 & 122706.66 & 333.27 & 1447.46 & 68023.54 & 153.42 & 607.5 \\
			cRVDC & 80 & 122706.66 & 332.24 & 1420 & 68023.54 & 148.07 & 582.97 \\
			\hline
			RVDC   & 128 & 94041.80 & 146.87 & 738.04 & 51771.48 & 76.93 & 299.02 \\
			cRVDC & 128 & 94041.80 & 161.45 & 701.09 &  51771.48 & 74.1 & 315.97  \\
			SPHINCS$^+$-128f & 128 & 194.81 & 12.88 & 143.73 & na & na & na \\
			\hline
			RVDC   & 192 & 47343.91 & 62 & 267.3 & 24982.4 & 32.79 & 130.31 \\
			cRVDC & 192 & 47343.91 & 64.69 & 287.27 & 24982.4 & 31.61 & 129.87\\
			SPHINCS$^+$-192f & 192 & 132.14 & 9.73 & 93.75 & na & na & na\\
			\hline
			RVDC   & 256 & 28134.23 & 43.49 & 178.53 & 14157.74 & 19.79 & 81.46\\
			cRVDC & 256 & 28134.23 & 41.74 & 182.27 & 14157.74 & 19.08 & 80.33 \\
			SPHINCS$^+$-256f & 256 & 55.72 & 4.7 & 95.45 & na & na & na\\
		\end{tabular}
	\end{center}
	\caption{RVDC and cRVDC operations per second.}
	\label{tab:perf}
\end{table}
%


\section{Conclusions}
\label{sec:concl}

We have presented two code-based signature schemes derived from a 5-pass identification protocol with cheating probability close to 1/2, using double circulant codes in the rank metric.
The second scheme optimizes the signature size from the first one, at the cost of few hash computations.
The resulting signature scheme has a signature size of approximately 11, 22, 54, and 93 KBytes for a corresponding security level of 96, 125, 193, and 254.
When compared to one of the most popular post-quantum hash-based signature schemes, namely SPHINCS+, the key generation algorithm is between 350 and 500 times faster, the signing algorithm is approximately ten times faster, and the verification algorithm is twice as fast.

%
%
%

%

\appendix

\section{Post-quantum security of the Fiat-Shamir transform}
\label{sec:pqc_fs}

It is well known that the Fiat-Shamir transform is secure in the random oracle model (ROM), see e.g. \cite{kiltz2018concrete}.
However, when the adversary has a quantum access to the oracle, i.e. in the quantum random oracle model (QROM), the situation is somehow more complex, and recently many results have been published (e.g. \cite{boneh2011random}, \cite{unruh2015non}, \cite{kiltz2018concrete}, \cite{unruh2017post}).
Since most of the schemes we compare to do not take into account this scenario,
we also omit it, and leave it to future research.

An alternative quantum secure transform by Unruh \cite{unruh2015non}
could be used instead of the Fiat-Shamir one,
yielding though a considerably less efficient signature,
since multiple executions of the underlying identification scheme are required.

In \cite{unruh2017post}, it is proven that if
a sigma-protocol has honest-verifier zero-knowledge and statistical soundness with a dual-mode hard instance generator,
then the resulting Fiat-Shamir signature scheme is unforgeable in the quantum scenario.
It is easy to see that our proposal has a dual-mode hard instant generator and honest-verifier (computational) zero-knowledge, but
on the other hand, a computationally unbounded adversary would prevent us from achieving statistical soundness. 
Thus, we cannot apply the results of \cite{unruh2017post} to our proposal.
Still, to the best of our knowledge, no quantum attack has been published to Veron-like constructions.

\section{Zero-knowlegde properties of RVDC signature scheme}
\label{sec:id_veron_dc_zn}
In this section we prove the security of RVDC scheme by showing how the completeness, soundness and zero-knowledge properties are achieved. In the proofs we follow \cite{aguilar2011new}.

\subsubsection{Completeness}
Given $(\sk,\pk)$ output from $\kgen$ function, 
it easy to see that for any possible $\sk = (x,e)$ the Verifier always accepts after interacting with the Prover $P$ on common input $\pk$. 
This is because the honest Prover who knows $\sk$ is be able to construct the three commitments $c_1,c_2,c_3$. 
Furthermore, the Verifier is always able to identify the Prover because the verifications match with the given commitments.

In particular, 
the check on the value $c_3$ when $b=0$ is valid because of Lemma \ref{thm:gamma_property}, i.e. $\Gamma_a(x \cdot G) = \Gamma'_a(x) \cdot G $. Thanks to this we have that 
$
u \cdot G + \Gamma_a(e) = 
u \cdot G + \Gamma_a(x \cdot G) + \Gamma_a(y) =
u \cdot G + \Gamma'_a(x) \cdot G + \Gamma_a(y) =
(u + \Gamma'_a(x)) \cdot G + \Gamma_a(y)
$.

Notice also that the components of the first challenge $a$ cannot be all the same,
otherwise $\rankweight(\Gamma_a(e)) = 0$ or $2$, 
depending of $a$ being equal to $(0, \ldots,0)$, $(\tilde{a}, \ldots,\tilde{a})$ respectively, 
and the check when $b=1$ would fail.

\subsubsection{Soundness}
We will show that if someone can be successfully identified by $\verifier$ with the protocol,
then it is able to retrieve the secret in polynomial time with a certain probability.
To do so, we introduce a specific problem which is easier to be solved than the syndrome decoding,
\footnote{This problem is the analog of the \emph{Differential Syndrome Decoding Problem} 
	(denoted \emph{Probl{\`e}me de d{\'e}codage par syndrome diff{\'e}rentiel}) in \cite{schrek2013signatures},
	for the Hamming metric.
	The same problem is used in \cite{aguilar2011new}.} 
except when there is only one solution, 
in which case the two problems are the same.
The way in which we assure the security is by 
choosing the parameters which allow to decrease 
the size of the solutions of the new problem to one with a probability exponentially close to 1 
(in practice, this probability to have more than one solution is $2^{\secpar}$).

\begin{problem}[Differential Rank Decoding Problem]
	\label{prob:DRDP}
	Consider $H$ a random double circulant matrix, 
	$Y$ a random codeword in $(\FF_q)^n$ of rank weight $r$, 
    and $A = \{a_1,\ldots,a_{\rho}\} \subseteq (\FF_{q})^k$, 
    with $a_j$ all distinct for $j=1,\ldots,\rho$, and
    $a_j=(\alpha_1,\ldots,\alpha_k)$,
    with $\alpha_1,\ldots,\alpha_k$ all distinct.
	Let $H\cdot Y^T$ be a syndrome.
	The problem ${\mathcal P}(H,Y\rho,A,r)$ consists in finding $\rho$ words $z_j$ and a constant $C$ 
	such that $H \cdot \Gamma_{a_j}(Y)^T - H \cdot z_j^T = C$,
	and $\rankweight(z_j) = r$ for all $j < \rho$.
\end{problem}

The above mentioned problem is easier than the
independent syndrome decoding problem, because of the addition of the unknown $C$.
However, it still seems to be hard to be solved.
Note that we can suppose that there exist a particular solution $Z_1, \ldots, Z_\rho, C$ to the problem ${\mathcal P}(H,Y\rho,A,r)$,
such that $C$ is equal to 0.
In this case, we have to solve the usual rank syndrome decoding problem 
$H \cdot  \Gamma_{a_j}(Y)^T= H \cdot z_j^T$ for all $j < i$.



Lemma \ref{thm:lemma2} gives the probability to find a solution of Problem \ref{prob:DRDP}
\begin{lemma}\label{thm:lemma2}
    Consider $\rho,A,r$ fixed.
	Let $Z_C = (Z_1, \ldots, Z_\rho, C)$ be a random vector with $Z_j, 1 \le j \le \rho$ a random variable with uniform distribution over the words of rank weight $r$,
    and $C$ a random variable with uniform distribution over $(\FF_{q^m})^{n-k}$.	
    Let $S_\rho$ be a random variable equal to the set of the solutions of the problem ${\mathcal P}(H,Y\rho,A,r)$,
    $\rho,A,r$ as in Problem \ref{prob:DRDP}.
    Note that $S_{\rho}$ is a random variable, 
    in the sense that $S_\rho$ is defined relatively to $H$ a random double circulant matrix and 
    $Y$ a random word of weight $r$.
	We have $\prob{Z_C \in S_\rho} =\frac{1}{(q^{m(n-k)})^\rho}$.
\end{lemma}
\begin{proof}
	\begin{align*}
	& \prob{
		Z_C \in S_\rho} = \\
	& = \prob{
		H \cdot Z_1^T = C + H \cdot \Gamma_{a_1}(X)^T
		\cap
		\ldots
		\cap
		H \cdot Z_\rho^T = C + H \cdot \Gamma_{a_\rho}(X)^T
	}
	\,,
	\end{align*}
	which, by the conditional probability formula, 
	is the product of the following two probabilities
	\begin{align*}
	& \prob{
		H \cdot Z_1^T = C + H \cdot \Gamma_{a_1}(X)^T 
		|
		\bigcap_{j=2}^{\rho} 
		H \cdot Z_j^T = C + H \cdot \Gamma_{a_j}(X)^T
	} 
	\cdot \\
	& \prob{
		\bigcap_{j=2}^{\rho}
		H \cdot Z_j^T = C + H \cdot \Gamma_{a_j}(X)^T
	}
	.
	\end{align*}
	In the case where the words of rank weight $r$ do not have a common image for $H$, we have 
	that:
	\begin{align*}
	 \prob{
		Z_C \in S_\rho} = 
	& \prob{
		H \cdot Z_1^T = C + H \cdot \Gamma_{a_1}(X)^T 
		|
		\bigcap_{j=2}^{\rho} Z_1 \ne Z_j
	} 
	\cdot \\
	 & \prob{
		\bigcap_{j=2}^{\rho}
		H \cdot Z_j^T = C + H \cdot \Gamma_{a_j}(X)^T
	}
	\,.
	\end{align*}
	These variables are independent, so
	\begin{align*}
	& \prob{
		Z_C \in S_\rho} = \\
	& = \prob{
		H \cdot Z_1^T = C + H \cdot \Gamma_{a_1}(X)^T
	} 
	\cdot
	\prob{
		\bigcap_{j=2}^{\rho}
		H \cdot Z_j^T = C + H \cdot \Gamma_{a_j}(X)^T
	}
	\,.
	\end{align*}
	Using a recursive argument, we have that 
	\begin{align*}
	& \prob{
		Z_C \in S_\rho} =
	\prob{
		H \cdot Z_1^T = C
	}^\rho
	\,.
	\end{align*}
	The hardness of the Decisional Rank Double Circulant Syndrome Decoding (DRDCSD) Problem (Defined in \cite{aguilar2016efficient})
	assures that the syndromes associated to codewords of given rank are indistinguishable from random syndromes, i.e. they are uniformly distributed among the syndrome space $(\FF_{q^m})^{n-k}$.
	Thus, we can conclude that
	$$
	\prob{Z_C \in S_\rho} = \frac{1}{(q^{m(n-k)})^\rho}.
	$$
	\qed
\end{proof}

\begin{lemma}\label{thm:lemma3}
	The distribution of $N_\rho$ describing the size of $S_\rho$ is the same of the variable $1 + Y$,
	with $Y$ a binomial distribution with parameters $N = (q^{m(n-k)} - 1)\genfrac[]{0pt}{}{n}{r}^\rho$ and
	$p = \frac{1}{(q^{m(n-k)})^\rho}$.
	Furthermore
	$$
	\mathbb{E}[N_\rho] = Np + 1 = (q^{m(n-k)} - 1) \left(\frac{\genfrac[]{0pt}{}{n}{r}}{q^{m(n-k)}}\right)^\rho \,.
	$$
\end{lemma}
\begin{proof}
	Let $Z_C = (Z_1, \ldots, Z_\rho, C)$
	be the random vector defined in Lemma \ref{thm:lemma2},
	with $C \ne 0$ and $T_C$ the variable equal to 1 when $Z_C \in S_\rho$ and 0 otherwise.
	$N_\rho = \sum_{C \ne 0} T_C + T_0$, with $T_0$ the number of solutions when $C = (0, \ldots,0)$.
	The variable $T_0$ is equal to 1 since
	for a given $C$ and $\rho$ distinct codewords of rank weight $r$
	only one solution can be found.
	The number of words of given rank weight $r$ is given by the 
	number of vector subspaces of length $n$ and dimension $r$, 
	which is indicated with $\genfrac[]{0pt}{}{n}{r}$ (defined in Section \ref{sec:prel}),
	while the number of all possible $C \ne (0, \ldots,0)$ is $q^{m(n-k)}  - 1$.
	So, we have
	$N_\rho = 1 + Y$ with $Y$ a binomial distribution with parameters 
	$N = (q^{m(n-k)}  - 1)\genfrac[]{0pt}{}{n}{r}^\rho$ and
	$p = \frac{1}{(q^{m(n-k)} )^\rho}$.
	\qed
\end{proof}

\begin{lemma}\label{thm:lemma4}
	Let $Y'$ be a random variable with Poisson distribution with parameter $Np$.
	Then we have
	$\prob{N_\rho = 1} \approx \prob{Y' = 0} \approx 1 - \frac{\genfrac[]{0pt}{}{n}{r}^\rho}{q^{m(n-k)(\rho-1)}}$.
\end{lemma}
\begin{proof}
	For a sufficiently large $N$ and sufficiently small $p$, 
	the binomial distribution of $Y'$ is approximated by the Poisson distribution
	with parameter $\lambda = Np$.
	We can deduce that the probability 
	$
	\prob{N_\rho = 1} 
	\approx 
	\prob{Y' = 0} = e^{-Np}\frac{(Np)^0}{0!}
	\approx 
	e^{-\frac{\genfrac[]{0pt}{}{n}{r}^\rho}{q^{m(n-k)(\rho-1)}}}
	$.
	When $x$ is closed to 0, we have that $e^x \approx 1 - x$, and thus
	$\prob{N_\rho = 1} \approx 1 - \frac{\genfrac[]{0pt}{}{n}{r}^\rho}{q^{m(n-k)(\rho-1)}}$.
	\qed
\end{proof}
Let us call $\epsilon$ the value $1 - \frac{\genfrac[]{0pt}{}{n}{r}^\rho}{q^{m(n-k)(\rho-1)}}$.

\begin{lemma}\label{thm:lemma5}
	If someone is able to solve the problem ${\mathcal P}(H,Y\rho,A,r)$ with probability $\epsilon'$,
	then he is also able to find the secret key of the protocol from the public key 
	with a probability of about $\epsilon\epsilon'$.
\end{lemma}
\begin{proof}
	We have from Lemma \ref{thm:lemma4} that the probability that the solution of ${\mathcal P}(H,Y\rho,A,r)$ is unique is $\epsilon$.
	\qed
\end{proof}

\begin{theorem}\label{thm:main1}
	If a prover $\prover$ is able to be authenticated by a a verifier $\verifier$ 
	with a probability greater than $\frac{q^k+\rho}{2q^k}$,
	then $\prover$ is able to retrieve the secret key of the protocol from the public key
	with a probability  greater than $1 - \frac{\genfrac[]{0pt}{}{n}{r}^\rho}{q^{m(n-k)(\rho-1)}}$ 
	in polynomial time or
	to find a collision on the underlying hash function in a polynomial time.
\end{theorem}
\begin{proof}
	The prover $\prover$ is able to correctly answer more than $k+\rho$ challenge queries.
	In this case, let us call $a,b$, respectively, the first and the second challenge of $\verifier$.
	First $\prover$ randomly chooses $P \in M_{n,n}(\FF_{q})$, $Q \in M_{m,m}(\FF_{q})$, and
	$v \in (\FF_{q^m})^n$,
	and sends the first two commitments
	$c_1=\hash(P,Q)$ and $c_2=\hash(v)$ to $\verifier$.
	
	We call $c_3$ the second commitment sent to $\verifier$.
	
	We also call $(u_a,P_a,Q_a)$ and $(v_a,z_a)$ the last response, respectively, 
	when $b=0$, and when $b=1$.
	For the Pigeonhole principle, $\prover$ is able to answer 
	to the challenge $(a,b=0)$ and $(a,b=1)$
	for at least $\rho$ different $a$, which we call $a_1, \ldots, a_\rho$.
	$\verifier$ must verify that 
	$\hash(P_{a_j},Q_{a_j}) = c_1$ and 
	$\hash(v_{a_j}) = c_2$.
	Thus, for any $j \in \{1, \dots, \rho\}$, either  $\prover$ finds a collision of the hash function, or 	
	$(P_{a_j},Q_{a_j}) = (P,Q)$ and $v_{a_j} = v$ for all $a_j$.
	$\verifier$ must also verify that the rank weight of $z_{a_j}$ equals $r$ and that the commitment $c_3$ is correct.
	To meet this last condition, 
	the values $P, Q,u_{a_j},v,z_{a_j}$ 
	generated by $\prover$ 
	must satisfy the condition 
	$
	\hash(\Pi_{P,Q}(u_{a_j}G + \Gamma_{a_j}(y))) = 
	\hash(v + z_{a_j})
	$, since both side of the equation must be equal to $c_3$.
	Thus, either  $\prover$ finds a collision of the hash function, or 	
	$\Pi_{P,Q}(u_{a_j}G + \Gamma_{a_j}(y)) = v + z_{a_j}$.
	In this case, we deduce that 
	$
	u_{a_j}G + \Gamma_{a_j}(y) = 
	\Pi_{P,Q}^{-1}(v) + \Pi_{P,Q}^{-1}(z_{a_j})$, 
	and then
	$
	H \cdot  \Gamma_{a_j}(y)^T - H \cdot \Pi_{P,Q}^{-1}(z_{a_j}) = 
	H \cdot \Pi_{P,Q}^{-1}(v)^T
	$.
	Since $H \cdot  \Gamma_{a_j}(y)^T = H \cdot  \Gamma_{a_j}(e)^T $,
	the previous equation corresponds to the problem ${\mathcal P}(H,Y\rho,A,r)$.
	We deduce from Lemma \ref{thm:lemma5} that $\prover$ is able to find the secret key with a probability greater than $\epsilon$.
	\qed
\end{proof}

\begin{theorem}\label{thm:main2}
	If a prover $\prover$ is able to be authenticated by a verifier $\verifier$ 
	with a probability greater than $\left( \frac{q^k+\rho}{2q^k} \right)^N$,
	then $\prover$ is able to retrieve the secret key of the protocol from the public key 
	with a probability greater than $1 - \frac{\genfrac[]{0pt}{}{n}{r}^\rho}{q^{m(n-k)(\rho-1)}}$ 
	in polynomial time or
	to find a collision on the underlying hash function in a polynomial time.
\end{theorem}
\begin{proof}
	$\prover$ is able to build $c_{1,1}, \ldots, c_{1,N}$ and $c_{2,1}, \ldots, c_{2,N}$
	such that it can be authenticated with a probability greater than 
	$\left( \frac{q^k+\rho}{2q^k} \right)^N$.
	For the Pigeonhole principle, we can deduce the existence of an integer $j$ 
	such that $\prover$ can be authenticated by the first protocol 
	with a probability greater than $\frac{q^k+\rho}{2q^k}$.
	Theorem \ref{thm:main1} allows to conclude the proof.
	\qed
\end{proof}

\subsubsection{Zero-Knowledge}
We need to prove that, beside the public parameters, no information can be deduced in polynomial time from an execution of the protocol.

We need to construct a polynomial time simulator $S$ of the protocol that, by interacting with the verifier $V$, provides a transcript which is indistinguishable from the one of the original protocol.

The simulator $S$ should perform the following steps
\begin{itemize}
	\item if $b=0$:
	\begin{itemize}
		\item choose random $P' \in M_{n,n}(\FF_q)$, $Q' \in M_{m,m}(\FF_q)$, and $v \in (\FF_{q^m})^n$;
		\item choose random $a' \in (\FF_{q})^{k}$;
		\item compute $h_1 = \hash(P',Q')$, and $h_3 = \hash(\Pi_{P',Q'}(v \cdot G + \Gamma_{a'}(y)))$.
	\end{itemize}
	Note that $P',Q',v$ are indistinguishable from $P,Q,u+\Gamma_{a'}(x)$;
	\item if $b=1$:
	\begin{itemize}
		\item choose random $P' \in M_{n,n}(\FF_q)$, $Q' \in M_{m,m}(\FF_q)$, $v \in (\FF_{q^m})^n$, and $z \in (\FF_{q^m})^n$ such that $\rankweight(z)=r$;
		\item compute $h_2 = \hash(\Pi_{P',Q'}(v))$, and $h_3 = \hash(\Pi_{P',Q'}(v) + z)$.
	\end{itemize}
	Note that $\Pi_{P',Q'}(v),z$ are indistinguishable from $\Pi_{P,Q}(u \cdot G), \Pi_{P,Q}(\Gamma_{a'}(e))$, since, if $P,Q$ are random matrices, then the function $\Pi_{P,Q}$ can map a vector of a certain rank to any vector of the same rank. Furthermore, the function $\Gamma_{a}$ preserves the rank.
\end{itemize}
The simulator just described runs in polynomial time.

Remark: 
notice that Stern-like schemes as usually presented
are only weak-testable zero knowledge (see \cite[Sect. 3.2]{alamelou2017code}).
They can be straightforwardly turned into a full ZK scheme following \cite[Theorem 3 and 4]{stern1996new}.

\section{Veron Double Circulant Signature schemes in the rank metric}
\label{sec:sign_descr_rvdc}

%
In this section we provide the description of the signature scheme 
derived from Veron identification protocol using double circulant codes (Section \ref{sec:id_veron_dc}),
which we will refer to as \emph{Rank Veron Double Circulant (RVDC) Signature scheme}.
We also consider a version of the scheme with a signature compression, and we refer to it as \emph{compressed Rank Veron Double Circulant (cRVDC)  signature scheme}.

Signature and verification algorithm for the RVDC and cRVDC schemes can be observed, respectively, in Fig.~\ref{fig:rvdcsignverifyalg} and Fig.~\ref{fig:crvdcsignverifyalg}. Key generation is the same as in Sect. \ref{sec:id_veron_dc}.

\restylefloat*{figure}
\begin{figure}[!ht]
	\centering
	\makebox[0pt][c]{%
		\hspace{-5cm}
		\begin{minipage}[t]{0.1\textwidth}
			\vspace{0pt}
			\setlength{\parindent}{0pt}
			\procedure[codesize=\scriptsize]{RVDC: $\sign(\sk, \pk, \msg, \delta)$}
			{%
				\sk = (x,e) \gets \kgen \\
				\pk = (y,G,r) \gets \kgen \\
				\msg \text{, message} \\
				\delta \text{, number of rounds as defined in Sect. \ref{sec:param}} \\
				[0.1\baselineskip][\hline]\< \< \\[-0.5\baselineskip]
				\pcln \pcfor i=1..\delta \pcdo \\
				\pcln \pcind u_i \sample (\FF_{q^m})^k \\
				\pcln \pcind P_i \sample M_{n,n}(\FF_q), Q_i \sample M_{m,m}(\FF_q) \\
				\pcln \pcind c_{i,1} \gets \hash(P_i, Q_i)\\
				\pcln \pcind c_{i,2} \gets \hash(\Pi_{P_i,Q_i}(u_i \cdot G ))\\
				\pcln \cmt_0 \gets c_{1,1} \| c_{1,2} \| \ldots \| c_{\delta,1} \| c_{\delta,2} \\
				\pcln \ch_1 \gets \hash(\cmt_0 \| \msg) \\ \label{step:or1aRVDC}
				\pcln \text{Truncate rightmost bits in } \ch_1 \\
				\pcind[2] \text{ so that it has } \delta k \log_2(q) \text{ bits} \\
				\pcind[2] \text{ and there is no block of length } k \log_2(q) \\
				\pcind[2] \text{ with all equal components over $\FF_q$. } \\ \label{step:or1b}
				%
				\pcln \pcfor i=1..\delta \pcdo \\
				\pcln \pcind a_i \gets \left(\ch_{1,(i-1)k \log_2(q)+1}, \ldots, \ch_{1,i k \log_2(q)}\right)  \\
				\pcln \pcind c_{i,3} \gets \hash \left(\Pi_{P_i,Q_i}\left(u_i \cdot G + \Gamma_{a_i}(e)\right)\right) \\
				\pcln \pcind \cmt_{1} \gets c_{1,3} \| \ldots \| c_{\delta,3} \\
				\pcln \ch_2 \gets \hash(\cmt_1) \\ \label{step:or2}
				\pcln \pcfor i=1..\delta \pcdo \\
				\pcln \pcind \pcif \ch_{2,i} = 0 \\
				\pcln \pcind[2] \rsp_{i,1} \gets (P_i, Q_i)\\
				\pcln \pcind[2] \rsp_{i,2} \gets u_i + \Gamma'_{a_i}(x)\\
				\pcln \pcind \pcif \ch_{2,i} = 1 \\
				\pcln \pcind[2] \rsp_{i,1} \gets \Pi_{P_i,Q_i}(u_i \cdot G)\\
				\pcln \pcind[2] \rsp_{i,2} \gets \Pi_{P_i,Q_i}(\Gamma_{a_i}(e))\\
				\pcln \signature \gets [\cmt_0, \cmt_1, \rsp] \\
				\pcln \pcreturn \signature\\
			}
		\end{minipage}
		\hspace{5cm}
		\begin{minipage}[t]{0.1\textwidth}
			\vspace{0pt}
			\procedure[codesize=\scriptsize]{RVDC: $\verif(\pk, \msg, \delta, \signature)$}
			{%
				\pk = (y,G,r) \gets \kgen \\
				\msg \text{, message} \\
				\delta \text{, number of rounds as defined in Sect. \ref{sec:param}} \\
				\signature = [\cmt_0, \cmt_1, \rsp] \text{, signature} \\
				[0.1\baselineskip][\hline]\< \< \\[-0.5\baselineskip]
				%
				\pcln \ch_1 \gets \hash(\cmt_0 \| \msg) \\ 
				\pcln \ch_2 \gets \hash(\cmt_1) \\ 
				\pcln \pcfor i=1..\delta \pcdo \\
				\pcln \pcind a_i \gets (\ch_{1,(i-1)k \log_2(q)+1}, \ldots, \ch_{1,i k \log_2(q)})  \\
				\pcln \pcind c_{i,3} \gets \cmt_{1,[\hlength(i-1)+1,\ldots,\hlength i]}\\
				\pcln \pcind \pcif \ch_{2,i} = 0 \\
				\pcln \pcind[2] c_{i,1} \gets \cmt_{0,[2\hlength(i-1)+1,\ldots,2\hlength(i-1)+\hlength]}\\
				\pcln \pcind[2] \pcif c_{i,1} \ne \hash(\rsp_{i,1}) \vee \\
				\pcln \pcind[3]c_{i,3} \ne \hash(\Pi_{\rsp_{i,1}}(\rsp_{i,2}\cdot G +  \Gamma_{a_i}(y))\\
				\pcln \pcind[3] \pcreturn \false \\
				\pcln \pcind \pcif \ch_{2,i} = 1 \\
				\pcln \pcind[2] c_{i,2} \gets \cmt_{0,[2\hlength(i-1)+\hlength+1,\ldots,2\hlength i]} \\
				\pcln \pcind[2] \pcif c_{i,2} \ne \hash(\rsp_{i,1}) \vee \\
				\pcln \pcind[3]c_{i,3} \ne \hash(\rsp_{i,1} + \rsp_{i,2}) \vee \rankweight(\rsp_{i,2}) \ne r \\
				\pcln \pcind[3] \pcreturn \false \\
				\pcln \pcreturn \true \\
			}
		\end{minipage}
	}
	\caption{RVDC signature and verification algorithms.}
	\label{fig:rvdcsignverifyalg}
\end{figure}

\restylefloat*{figure}
\begin{figure}[!ht]
	\centering
	\makebox[0pt][c]{%
		\hspace{-5cm}
		\begin{minipage}[t]{0.1\textwidth}
			\vspace{0pt}
			\setlength{\parindent}{0pt}
			\procedure[codesize=\scriptsize]{cRVDC: $\sign(\sk, \pk, \msg, \delta)$}
			{%
				\sk = (x,e) \gets \kgen \\
				\pk = (y,G,r) \gets \kgen \\
				\msg \text{, message} \\
				\delta \text{, number of rounds as defined in Sect. \ref{sec:param}} \\
				[0.1\baselineskip][\hline]\< \< \\[-0.5\baselineskip]
				\pccomment{Step 1} \\
				\pcln \pcfor i=1..\delta \pcdo \\
				\pcln \pcind u_i \sample (\FF_{q^m})^k \\
				\pcln \pcind \seed_i \sample \{0, \ldots, 2^\lambda-1\} \\
				\pcln \pcind P_i \gets \xof(1,\seed_i)_{n^2} \\ 
				\pcln \pcind Q_i \gets \xof(2,\seed_i)_{m^2} \\
				\pcln \pcind c_{i,1} \gets \xof(P_i, Q_i)_{2\lambda}\\
				\pcln \pcind c_{i,2} \gets \xof(\Pi_{P_i,Q_i}(u_i \cdot G ))_{2\lambda}\\
				\pcln \cmt_0 \gets \xof(c_{1,1} \| c_{1,2} \| \ldots \| c_{\delta,1} \| c_{\delta,2})_{2\lambda} \\
				\pccomment{Step 2} \\
				\pcln \ch_1 \gets \xof(\cmt_0 \| \msg)_{\delta k \log_2(q)} \\ \label{step:or1acRVDC}
				%
				\pccomment{Step 3} \\
				\pcln \pcfor i=1..\delta \pcdo \\
				\pcln \pcind a_i \gets \left(\ch_{1,(i-1)k \log_2(q)+1}, \ldots, \ch_{1,i k \log_2(q)}\right)  \\
				\pcln \pcind c_{i,3} \gets \xof \left(\Pi_{P_i,Q_i}\left(u_i \cdot G + \Gamma_{a_i}(e)\right)\right)_{2\lambda} \\
				\pcln \pcind \cmt_{1} \gets c_{1,3} \| \ldots \| c_{\delta,3} \\
				\pccomment{Step 4} \\
				\pcln \ch_2 \gets \xof(\cmt_1)_{2\lambda} \\ 
				\pccomment{Step 5} \\
				\pcln \pcfor i=1..\delta \pcdo \\
				\pcln \pcind \pcif \ch_{2,i} = 0 \\
				\pcln \pcind[2] \rsp_{i,1} \gets u_i + \Gamma'_{a_i}(x) \\
				\pcln \pcind[2] \rsp_{i,2} \gets \seed_i \\
				\pcln \pcind[2] \rsp_{i,3} \gets c_{i,2}\\
				\pcln \pcind \pcif \ch_{2,i} = 1 \\
				\pcln \pcind[2] \rsp_{i,1} \gets \Pi_{P_i,Q_i}(u_i \cdot G)\\
				\pcln \pcind[2] \rsp_{i,2} \gets \text{Coordinates of }\Pi_{P_i,Q_i}(\Gamma_{a_i}(e))\\
				\pcln \pcind[2] \rsp_{i,3} \gets c_{i,1}\\
				\pcln \signature \gets [\cmt_0, \cmt_1, \rsp] \\
				\pcln \pcreturn \signature\\
			}
		\end{minipage}
		\hspace{5cm}
		\begin{minipage}[t]{0.1\textwidth}
			\vspace{0pt}
			\procedure[codesize=\scriptsize]{cRVDC: $\verif(\pk, \msg, \delta, \signature)$}
			{%
				\pk = (y,G,r) \gets \kgen \\
				\msg \text{, message} \\
				\delta \text{, number of rounds as defined in Sect. \ref{sec:param}} \\
				\signature = [\cmt_0, \cmt_1, \rsp] \text{, signature} \\
				[0.1\baselineskip][\hline]\< \< \\[-0.5\baselineskip]
				%
				\pcln \ch_1 \gets \xof(\cmt_0 \| \msg)_{2\lambda} \\ 
				\pcln \ch_2 \gets \xof(\cmt_1)_{2\lambda} \\ 
				\pcln \pcfor i=1..\delta \pcdo \\
				\pcln \pcind a_i \gets (\ch_{1,(i-1)k \log_2(q)+1}, \ldots, \ch_{1,i k \log_2(q)})  \\
				\pcln \pcind c_{i,3} \gets \cmt_{1,[\hlength(i-1)+1,\ldots,\hlength i]}\\
				\pcln \pcind \pcif \ch_{2,i} = 0 \\
				\pcln \pcind[2] P' \gets \xof(1,r_{i,2})_{n^2} \\
				\pcln \pcind[2] Q' \gets \xof(2,r_{i,2})_{m^2} \\ 
				\pcln \pcind[2] c_{i,1} \gets \xof(P',Q')_{2\lambda} \\
				\pcln \pcind[2] c_{i,2} \gets \rsp_{i,3} \\
				\pcln \pcind[2] \pcif c_{i,3} \ne \xof(\Pi_{P',Q'}(r_{i,1} \cdot G +  \Gamma_{a_i}(y))_{2\lambda}\\
				\pcln \pcind[3] \pcreturn \false \\
				\pcln \pcind \pcif \ch_{2,i} = 1 \\
				\pcln \pcind[2] c_{i,1} \gets \rsp_{i,3} \\
				\pcln \pcind[2] c_{i,2} \gets \xof(\rsp_{i,1})_{2\lambda} \\
				\pcln \pcind[2] \pcif c_{i,3} \ne \xof(\rsp_{i,1} + \rsp_{i,2}) \vee \rankweight(\rsp_{i,2})_{2\lambda} \ne r \\
				\pcln \pcind[3] \pcreturn \false \\
				\pcln \pcif \cmt_0 = \xof(c_{1,1} \| c_{1,2} \| \ldots \| c_{\delta,1} \| c_{\delta,2})_{2\lambda} \\
				\pcln \pcind \pcreturn \true \\
			}
		\end{minipage}
	}
	\caption{cRVDC signature and verification algorithms.}
	\label{fig:crvdcsignverifyalg}
\end{figure}

In the algorithm in Fig.~\ref{fig:rvdcsignverifyalg}, if $\delta k \log_2(q)$ is greater than $\hlength$, then it is possible to compute the challenge as $\ch \gets \hash(\cmt \| \msg \| 1) \| \ldots \| \hash(\cmt \| \msg \| l) \in (\FF_2)^{l \cdot \hlength}$,
where
$l \gets \left\lfloor \delta k \log_2(q) / \hlength \right\rfloor + 1$. 
Alternatively, one may use an Extended Output Function (XOF), as shown in Fig.~\ref{fig:crvdcsignverifyalg}, where $\xof(x)_l$ means that we take $l$ bits from the output of the function $\xof$ with input $x$.

\end{document}